\documentclass[conference]{IEEEtran}
\usepackage{amsfonts,amsmath,amsthm}
\usepackage[hidelinks]{hyperref}
\usepackage{enumerate,tikz}

\newcommand{\C}{\mathbb C}
\newcommand{\R}{\mathbb R}
\newcommand{\Z}{\mathbb Z}
\newcommand{\F}{\mathbb F}
\newcommand{\mcC}{\mathcal C}

\newcommand{\defeq}{\mathrel{\mathop :} =}
\newcommand{\eqdef}{= \mathrel{\mathop :}}

\newcommand{\tvec}[1]%
{ \big( \begin{smallmatrix} #1 \end{smallmatrix} \big) }
\newcommand{\svec}[1]%
{ {\small \begin{pmatrix} #1 \end{pmatrix}} }

\DeclareMathOperator{\re}{Re}
\DeclareMathOperator{\im}{Im}

\let\eps\varepsilon
\let\phi\varphi

\newtheorem{thm}{Theorem}
\newtheorem{prop}{Proposition}
\newtheorem{lem}{Lemma}
\newtheorem{fact}{Fact}

\theoremstyle{remark}
\newtheorem{rem}{Remark}
\newtheorem{exa}{Example}

\theoremstyle{definition}
\newtheorem{defn}{Definition}

\title{Codes Correcting Few Restricted Errors}

\author{\IEEEauthorblockN{Jens Zumbrägel}
  \IEEEauthorblockA{%
    \textit{Faculty of Computer Science and Mathematics}\\
    \textit{University of Passau}\\
    Innstraße 33, 94032 Passau, Germany}
    jens.zumbraegel@uni-passau.de}

\begin{document}

\maketitle

\begin{abstract}
  We consider linear codes over a field in which the error values are
  restricted to a subgroup of its unit group.  This scenario captures
  Lee distance codes as well as codes over the Gaussian or Eisenstein
  integers.  Codes correcting restricted errors gained increased
  attention recently in the context of code-based cryptography.
  
  In this work we provide new constructions of codes over the Gaussian
  or Eisenstein integers correcting two or three errors.  We adapt
  some techniques from Roth and Siegel's work on codes for the Lee
  metric.  We propose two construction methods, which may be seen of
  geometric and algebraic flavor, respectively.
\end{abstract}

\section*{Introduction}

Recently, the study of linear codes over a finite field~$\F_p$ with a
restricted error model was initiated~\cite{restricted}.  In this
scenario the error values are confined to a subset~$E$ of the base
field, typically a subgroup of its unit group~$\F_p^*$.  Since the
cases $E \defeq \F_p^*$ and $E \defeq \{ \pm 1 \}$ capture the Hamming
and the Lee error model, respectively, it is interesting to study the
characteristics of these codes for intermediate subgroups~$E$.  One
motivation for studying restricted errors comes from code-based
cryptography, with the proposal of efficient cryptosystems%
~\cite{freuden1, freuden2} and signatures~\cite{cross}.  While generic
decoding of restricted errors has been investigated~\cite{generic},
the construction of good codes and algebraic decoding methods appear
to be an intricate problem.

In this regard, codes over the Gaussian integers and the subgroup
$E \defeq \{ \pm 1, \pm i \}$ of order~$4$ gained some prominence in
the context of two-dimensional modulation schemes~\cite{gauss}, as did
codes over the Eisenstein ring with its order~$6$ subgroup~\cite{eisen}.
It was noted however that in both cases a straightforward construction
of constacyclic codes fails to correct two errors or more.  Indeed, it
seems difficult to adapt Berlekamp's decoding algorithm for negacyclic
codes in the Lee metric~\cite{berlekamp}, although some work in this
direction was done by Huber~\cite{huber}.

In the present contribution we adapt instead the Lee metric code
construction by Roth and Siegel~\cite{rothsiegel}.  By applying
techniques to circumvent the issue of error values beyond $\pm 1$, we
are able to construct codes correcting few errors in the case of a
subgroup~$E$ of order~$4$ or~$6$.  These codes have better parameters
for the restricted weight than maximum distance separable (MDS) codes
with the Hamming weight.  We also mention methods and obstacles when
extending the techniques to correct more errors.

The subsequent section covers preliminaries on the restricted weight
and the case of one-error correcting codes.  Then we turn to codes
over the Gaussian integers and propose code constructions for
correcting two or three errors.  Afterwards we deal with codes over
the Eisenstein ring and present two code constructions, one being of
geometrical flavor and the other being algebraic.

\section{Preliminaries}

Here we collect some basic results on the restricted weight, one-error
correcting codes and Newton identities.

\subsection{The restricted weight}

Let~$X$ be an abelian group.  By a \emph{weight} on the group~$X$ we
mean a function $w \colon X \to \R_{\ge 0}$ satisfying
\begin{enumerate}[\quad i)]
\item $w(x) = 0 ~\Leftrightarrow~ x = 0$,
\item $w(-x) = w(x)$,
\item $w(x + y) \le w(x) + w(y)$ \,(triangle inequality).
\end{enumerate}
Such a weight defines a metric $d \colon X \times X \to \R_{\ge 0}$ by
\[ d(x, y) \,\defeq\, w(x - y) \,. \]
As usual the weight and metric extend to~$X^n$ additively by $d(x, y)
\defeq \sum d(x_i, y_i) = \sum w(x_i - y_i)$ for $x, y \in X^n$.

\begin{exa} Two well-known distances in coding theory are the Hamming
  and Lee metrics.  The former is induced by the Hamming weight, which
  is given by $w(0) \defeq 0$ and $w(x) \defeq 1$ for $x \ne 0$.  For
  the latter we assume $X = \Z_q = \{ 0, \dots, q \!-\! 1 \}$,
  the integers modulo~$q$, and define for $x \in \Z_q$ the Lee weight
  by $w(x) \defeq \min \{ x, q \!-\! x \}$. \end{exa}

Now suppose that~$p$ is an odd prime and~$E$ is a subgroup
of~$\F_p^*$, where~$\F_p$ is the finite field with~$p$ elements.  We
let $m \defeq |E|$ be the order of~$E$, thus $m \mid p \!-\! 1$.  We
also assume $-1 \in E$, so that $E = -E$ and~$m$ is even.

\begin{defn} The \emph{restricted weight} on~$\F_p$ with respect to~$E$
  is for $x \in \F_p$ given by \[ w(x) \,\defeq\, \min \{ \ell \mid
    \eps_1 + \ldots + \eps_{\ell} = x ,\, \eps_i \in E \} \,, \]
  i.e.\ the shortest length of a path from~$0$ to~$x$ using errors
  in~$E$. \end{defn}

Note that the Hamming and Lee weights can be viewed as restricted
weights for $E \defeq \F_p^*$ and $E \defeq \{ \pm 1 \}$, respectively.

\begin{rem} A slightly different weight is defined in~\cite{restricted}
  by $w'(0) \defeq 0$, $w'(\eps) \defeq 1$ for $\eps \in E$,
  $w'(x) \defeq 2$ when $0 \ne x \in (E + E) \setminus E$, and
  $w'(x) \defeq \infty$ otherwise.  However in general, it does not
  fulfill the triangle inequality and induces no metric.  Our weight
  satisfies $w \le w'$, so the lower bounds on the minimum distance
  and the error-correction capability also hold for the
  weight~$w'$. \end{rem}

We are mainly interested in the cases $m = 4$ and $m = 6$,
corresponding to codes over Gaussian and Eisenstein integers,
respectively.

\subsection{One-error correcting codes}

Fix a weight~$w$ on an abelian group~$X$.  Consider an additive
code~$\mcC$ of length~$n$ over~$X$, i.e.\ a subgroup of~$X^n$, and
let \[ d \,\defeq\, \min \{ w(c) \mid 0 \ne c \in \mcC \} \] be its
minimum distance.  The code is \emph{$t$-error correcting} if
$d > 2 t$, which means that the balls $B_t(c)$ of radius~$t$ around
codewords $c \in \mcC$ are disjoint.  A $t$-error correcting code is
\emph{perfect} if these balls cover the whole space~$X^n$.

There always exist perfect one-error correcting codes for the
restricted weight.  Indeed, let~$E$ be a subgroup of~$\F_p^*$ of even
order~$m$.  Consider the equivalence classes on $\F_p^r \setminus
\{ 0 \}$, where $x \sim y$ if and only if $y = \eps x$ for some
$\eps \in E$, so that there are $n \defeq \frac 1 m (p^r \!-\! 1)$
classes.  Then define the linear code~$\mcC$ over~$\F_p$ by the
parity-check matrix
\[ H \,\defeq\, \begin{bmatrix}
    \,\mid & & \mid\, \\
    \,v_1 & \dots & v_n\, \\
    \,\mid & & \mid\, \end{bmatrix} , \]
where the $v_1, \dots, v_n \in \F_p^r \setminus \{ 0 \}$ form a set of
representatives of the classes.  A standard argument shows the following.

\begin{prop} The code~$\mcC$ is perfect one-error correcting. \end{prop}

\begin{exa} Let $p = 13$ and $m = 4$ so that $E = \{ \pm 1, \pm 5 \}$.
  We get for $m = 1$ the perfect one-error code defined by \[ H
    \,\defeq\, \begin{bmatrix} \,1 & 2 & 4\, \end{bmatrix} . \] \end{exa}

Golomb and Welch~\cite{golombwelch} suggest that for Lee metric these
are essentially the only perfect codes.  It would be intriguing to
find other perfect codes for restricted weights, cf.~\cite{mbg}.

\subsection{Newton identities}

An important tool for establishing bounds on the minimum distance and
decoding algorithms are the Newton identities.

\begin{fact}\label{fact:newton} Let~$F$ be a field.  The polynomial
  $\sigma \defeq \prod_j (1 - \alpha_j z)$, where $\alpha_1, \dots,
  \alpha_t \in F$, is related to the power series $S \defeq \sum_{i \ge 1}
  S_i z^i$ with coefficients $S_i \defeq \sum_j \alpha_j^i$ as
  \[ S \sigma + z \sigma' = 0 \,. \] \end{fact}

\begin{proof} Using the Leibniz rule $\sigma' = \sum_j - \alpha_j
  \frac \sigma {1 - \alpha_j z}$ and the power series $\frac {\alpha_j z}
  {1 - \alpha_j z} = \sum_{i \ge 1} (\alpha_j z)^i$ we may compute
  $-z \sigma' = \sum_j \alpha_j z \frac \sigma {1 - \alpha_j z} =
  \sigma \sum_j \sum_{i \ge 1} (\alpha_j z)^i = \sigma S$. \end{proof}

For correction of~$t$ errors at positions $\alpha_1, \dots, \alpha_t$
the polynomial~$\sigma$ and the coefficients~$S_i$ correspond to the
error locators and the syndromes, respectively.  Writing $\sigma =
1 + \sum_{j=1}^t a_j z^j$ we have $\sigma' = \sum_j j a_j z^{j-1}$ and
can thus reconstruct the coefficients by the Newton identities as
\begin{align*} - a_1 &= S_1 \\
  -2 a_2 &= S_2 + a_1 S_1 \\
  -3 a_3 &= S_3 + a_1 S_2 + a_2 S_1 \\ &\cdots \end{align*}
Hence, in particular if $t < p$ the characteristic of~$F$, then given
syndromes $S_1, \dots, S_t$ determine uniquely the coefficients
$a_1, \dots, a_t$ of the error locator polynomial~$\sigma$.

\section{Codes over the Gaussian integers}

We present a construction of linear codes having minimum distance
$\ge 6$, therefore correcting two restricted errors.  Instead of
adapting Berlekamp's negacyclic codes it is based on the Lee-error
correcting BCH codes devised by Roth and Siegel.  Afterwards we
modify our construction in order to correct three errors.  We start by
collecting some facts on codes over the Gaussian integers and the
Mannheim distance.

\begin{fact} Every prime $p \equiv 1 \pmod 4$ can be written as
  $p = a^2 + b^2$ with positive integers $a, b$, unique up to
  order. \end{fact}

Note that exactly one of~$a$ and~$b$ has to be even and the other is
odd.  Let $\pi \defeq a + i b \in \Z[i]$.

We consider the isomorphism $\F_p \cong \Z[i] / (\pi)$.  The ideal
$(\pi)$ can be viewed as a lattice~$L$ in $\Z[i]$ spanned by the
integral basis $(\pi, i \pi)$.  Seen as a lattice in $\Z^2$ the basis
vectors are $\tvec{a \\ b}$ and $\tvec{-b \\ a}$, which are orthogonal.
On $\Z[i]$ we have the Manhattan distance defined by the weight
\[ W(x + i y) \,\defeq\, |x| + |y| \,. \]
It induces the \emph{Mannheim distance}~\cite{gauss} on $\Z[i] /
(\pi)$ as a quotient distance modulo the lattice $L = (\pi)$, hence
\[ w([x + i y]) \,\defeq\, \min \{ W(x' + i y') \mid x' + i y'
  \in [x + i y] \} \,. \]
This defines a weight, where $w([x + i y])$ can be interpreted as the
length of the shortest path from~$0$ to a representative $x' + i y'$
using the steps $+1, +i, -1, -i$.

Let $R \defeq \{ x + i y \mid w([x + i y]) = W(x + i y) \}$ denote
the “fundamental region” of the lattice~$L$, which also equals
$R = ([-\frac 1 2, \frac 1 2] \pi + [-\frac 1 2, \frac 1 2] i \pi)
\cap \Z[i]$.

\begin{exa} Consider $p = 13 = 3^2 + 2^2$ so that $\pi = 3 + 2 i$.
  The field $\F_p = \{ -6, \dots, 6 \}$ can be represented as
  $\Z[i] / (\pi)$ with fundamental region and the Mannheim distance
  as follows.
  \begin{center}\begin{tikzpicture}[scale=0.6]
      \node at (0, 0) {$0$};
      \node at (1, 0) {$1$}; \node at (-1, 0) {$-1$};
      \node at (2, 0) {$2$}; \node at (-2, 0) {$-2$};
      \node at (0, -2) {$3$}; \node at (0, 2) {$-3$};
      \node at (-1, 1) {$4$}; \node at (1, -1) {$-4$};
      \node at (0, 1) {$5$}; \node at (0, -1) {$-5$};
      \node at (1, 1) {$6$}; \node at (-1, -1) {$-6$};
    \end{tikzpicture}\end{center} \end{exa}

\begin{fact} We have $\min \{ W(x + i y) \mid 0 \ne x + i y \in
  (\pi) \} = a + b$. \end{fact}

\begin{proof} This follows since the basis $(\pi, i \pi)$ of the
  lattice~$L$ is orthogonal with $W(\pi) = W(i \pi) = a + b$. \end{proof}

\subsection{Codes correcting two errors}

The code constructions and its distance property will depend on the
ability to perform “real” and “imaginary” parts in~$\F_p$ when
represented by Gaussian integers.  For this we choose some set
$A \subseteq \Z$ such that $A + i A \subseteq R$, the fundamental
region of the lattice $L = (\pi)$.

This allows to define real and imaginary parts of elements in
$\F_p \cong \Z[i] / (\pi)$ in an unambiguous way if $\alpha \in
A + i A$ by
\begin{gather*}
  \re \colon A + i A \to A \,, \quad x + i y \mapsto x \,, \\
  \im \colon A + i A \to A \,, \quad x + i y \mapsto y \,.
\end{gather*}
For $\alpha \in A + i A$ there always holds that $\alpha = \re \alpha
+ i \im \alpha$.  Note however that rules like $\re(\alpha + \beta)
= \re \alpha + \re \beta$ and $\im(\alpha + \beta) = \im \alpha +
\im \beta$ do not hold without further restriction of the domain
$A + i A$.

Let $n \defeq |A|^2$ and $J \defeq A + i A \subseteq R$ so that $|J| = n$.
Denote its elements as $J \eqdef \{ \alpha_1, \dots, \alpha_n \}$.  We
consider the linear code~$\mcC$ over~$\F_p$ of length~$n$ defined by
the parity-check matrix
\[ H \,\defeq\, \begin{bmatrix} 1 & \dots & 1 \\
    \re \alpha_1 & \dots & \re \alpha_n \\
    \im \alpha_1 & \dots & \im \alpha_n \\
    \alpha_1^2 & \dots & \alpha_n^2 \end{bmatrix} . \]

\begin{exa} In the case $p = 13$ we can use $A = \{ -1, 0, 1 \}$,
  so that $n = 9$ and $J = \{ -6, -5, -4, -1, 0, 1, 4, 5, 6 \}$.  The
  parity-check matrix is given by
  \[ H \,\defeq\, \begin{bmatrix}
      \,1 & 1 & 1 & 1 & 1 & 1 & 1 & 1 & 1\, \\
      -1\! & 0 & 1 & \!-1\! & 0 & 1 & \!-1\! & 0 & 1\, \\
      -1\! & \!-1\! & \!-1\! & 0 & 0 & 0 & 1 & 1 & 1\, \\
      -3\! & \!-1\! & 3 & 1 & 0 & 1 & 3 & \!-1\! & \!-3 \end{bmatrix} . \]
  We can infer from the argument below that the minimum distance is
  $\ge 5$ (since $a + b = 5$), thus correcting two errors. \end{exa}

Let us assume that the prime~$p$ satisfies $a + b \ge 7$, which occurs
if $p \ge 29$.

\begin{thm}\label{thm:first} The above code~$\mcC$ has minimum
  Mannheim distance $\ge 6$. \end{thm}

\begin{proof} (Inspired from Roth-Siegel~\cite{rothsiegel}.)  Suppose~$c$
  to be a codeword of Mannheim distance $< 6$, then we show that
  necessarily $c = 0$.  Denoting the rows of~$H$ by $h^1, \dots, h^4$
  we have $h^i \, c^{\top} = 0$ for $i = 1, \dots, 4$.

  The equation $h^1 \, c^{\top} = 0$ means that the values of the
  restricted errors sum up to~$0$.  (As $a + b \ge 7$ we cannot reach
  another lattice point.)  This allows to write $c = c_+ - c_-$ with
  $c_+, c_-$ having errors $+1$, $+i$ only, where~$c_+$ contains the
  errors $+1$, $+i$ and~$c_-$ accounts for the errors~$-1$, $-i$.  The
  number of errors in~$c_+$ and~$c_-$ is the same, thus the total
  weight of~$c$ is even and $\le 4$.  Now we consider two cases.
  \begin{enumerate}[\quad i)]
  \item Only~$+1$ and~$-1$ errors occur.  Then~$c_+$ and~$c_-$ contain
    at most two errors each.  Suppose the error positions are
    $\alpha, \beta$ and $\alpha', \beta'$, respectively.  Combining
    the second and third equations we get $\alpha + \beta = \alpha'
    + \beta'$.  Together with $\alpha^2 + \beta^2 = \alpha'^2 + \beta'^2$
    we can deduce $\alpha \beta = \alpha' \beta'$ and thus obtain the
    set $\{ \alpha, \beta \} = \{ \alpha', \beta' \}$.  The same
    argument applies if there are only~$+i$ and~$-i$ errors.
  \item Each of $c_+, c_-$ has one~$+1$ error and one~$+i$ error.
    Suppose that~$\alpha, \beta$ are their respective positions
    in~$c_+$, and similarly $\alpha', \beta'$ in~$c_-$.  Then we infer
    from the second row
    \[ \re \alpha + i \re \beta \,=\, \re \alpha' + i \re \beta' \]
    from which we obtain $\re \alpha = \re \alpha'$ and $\re \beta
    = \re \beta'$.  Likewise we obtain from the third row $\im \alpha
    = \im \alpha'$, $\im \beta = \im \beta'$, from which we get
    $\alpha = \alpha'$, $\beta = \beta'$.
  \end{enumerate}
  Hence in both cases we have that the error positions in~$c_+$ and~$c_-$
  are the same, which cannot occur unless $c = 0$.
\end{proof}

\begin{exa} For $p = 29$ we can take $A = \{ -2, -1, 0, 1, 2 \}$ and
  have $n = 25$ to obtain a code with minimum Mannheim distance
  $\ge 6$. \end{exa}

\begin{rem} A decoding method for correcting up to two restricted
  errors can be sketched as follows.  Suppose that~$x$ is a received
  vector containing errors at positions~$\alpha$ and~$\beta$ with
  values in $\{ \pm 1, \pm i \}$.  The first syndrome $s \defeq h^1
  x^{\top}$ then provides the sum of the error values.
  \begin{enumerate}[\quad i)]
  \item If $s = \pm 2$ or $s = \pm 2 i$, we get from the other
    syndromes $\alpha + \beta$, $\alpha^2 + \beta^2$, so we
    deduce $\alpha \beta$ and thus $\{ \alpha, \beta \}$.
  \item If $s = 0$, we have $\alpha - \beta$, $\alpha^2 - \beta^2$
    or $i \alpha - i \beta$, $i \alpha^2 - i \beta^2$, from which
    we compute $\alpha + \beta$.  This way we obtain~$\alpha$
    (from $\alpha - \beta$ or $i \alpha - i \beta$) and hence~$\beta$.
  \item Finally, if $s = 1 + i$, say, we have $\re \alpha + i \re \beta$
    and $\im \alpha + i \im \beta$, from which we directly get~$\alpha$
    and~$\beta$.
  \end{enumerate}\end{rem}

\subsection{Codes correcting three errors}

We can modify the above argument in order to construct codes of
minimum distance $\ge 8$, so being able to correct three restricted
errors.  For this we require a set $A \subseteq \Z$ satisfying
$-A = A$ and $A + A + i A \subseteq R$.  We also assume the prime~$p$
to be such that $a + b \ge 9$, which occurs if $p \ge 41$.

As before we let $n \defeq |A|^2$ and $J = \{ \alpha_1, \dots,
\alpha_n \} \defeq A + i A$.  We consider the code~$\mcC$ over~$\F_p$
given by the parity-check matrix
\[ H \,\defeq\, \begin{bmatrix} 1 & \dots & 1 \\
    \re \alpha_1 & \dots & \re \alpha_n \\
    \im \alpha_1 & \dots & \im \alpha_n \\
    \alpha_1^2 & \dots & \alpha_n^2 \\
    \alpha_1^3 & \dots & \alpha_n^3 \end{bmatrix} . \]

\begin{thm} The above code~$\mcC$ has minimum Mannheim distance
  $\ge 8$. \end{thm}

\begin{proof} We follow the outline of the previous proof.  So if~$c$
  is a codeword of Mannheim distance $< 8$, then as before we can
  write $c = c_+ + c_-$ with $c_+$ and $c_-$ having disjoint support
  and the same number of errors $+1$, $+i$.  In particular, the weight
  of~$c$ is even and $\le 6$.  We again have two cases.
  \begin{enumerate}[\quad i)]
  \item Only~$+1$ and~$-1$ errors occur.  Then~$c_+$ and~$c_-$ contain
    at most three errors each, say at positions $\alpha, \beta, \gamma$
    and $\alpha', \beta', \gamma'$.  From the parity-check equations we
    obtain $\alpha^i + \beta^i + \gamma^i = \alpha'^i + \beta'^i +
    \gamma'^i$ for $i = 1, 2, 3$, which allows to deduce the set
    $\{ \alpha, \beta, \gamma \} = \{ \alpha', \beta', \gamma' \}$
    using the Newton identities, cf.~Fact~\ref{fact:newton}.  The
    case of~$+i$ and~$-i$ errors is similar.
  \item Each of~$c_+$ and~$c_-$ has at most two~$+1$ errors and
    one~$+i$ error.  Suppose that in~$c_+$ the~$+1$ errors are at
    positions $\alpha, \beta$ and the~$+i$ error at position~$\gamma$,
    while $\alpha', \beta'$ and $\gamma'$ are the respective error
    positions in~$c_-$.  The second parity-check then gives
    \[ \re \alpha + \re \beta + i \re \gamma \,=\, \re \alpha' + \re
      \beta' + i \re \gamma' \,, \] from which by our assumption
    $A + A + i A \subseteq R$ we infer $\re \gamma = \re \gamma'$.
    Likewise we deduce from the third parity-check that $\im \gamma =
    \im \gamma'$, and hence $\gamma = \gamma'$.  But  then we also have
    $\re \alpha + \re \beta = \re \alpha' + \re \beta'$ and
    $\im \alpha + \im \beta = \im \alpha' + \im \beta'$, which
    implies that $\alpha + \beta = \alpha' + \beta'$.  Using the fourth
    parity-check equation we also obtain $\alpha^2 + \beta^2 =
    \alpha'^2 + \beta'^2$ and thus recover the set $\{ \alpha, \beta \}
    = \{ \alpha', \beta' \}$.  A similar argument applies in
    case~$c_+$ has two~$+i$ errors and one~$+1$ error.
  \end{enumerate}
  Therefore, in all cases the error positions in~$c_+$ and~$c_-$ coincide,
  which only is possible if $c = 0$.
\end{proof}

\begin{exa} For $p = 29$ the choice $A \defeq \{ -1, 0, 1 \}$ fulfills
  $A + A + i A \subseteq R$.  We obtain a code of length $n = 9$ with
  $J = \{ -13, -12, -11, -1, 0, 1, 11, 12, 13 \}$, which by the above
  argument has minimum distance $\ge 7$, thus correcting three
  restricted errors.  The parity-check matrix is given by
  \[ H \,\defeq\, \begin{bmatrix}
      \,1 & 1 & 1 & 1 & 1 & 1 & 1 & 1 & 1\, \\
      -1\! & 0 & 1 & \!-1\! & 0 & 1 & \!-1\! & 0 & 1\, \\
      -1\! & \!-1\! & \!-1\! & 0 & 0 & 0 & 1 & 1 & 1\, \\
      -5\! & \!-1\! & 5 & 1 & 0 & 1 & 5 & \!-1\! & \!-5 \\
      \,7 & \!12\! & 3 & \!-1\! & 0 & 1 & \!-3\! & \!\!-12\!\! & \!-7
    \end{bmatrix} . \] \end{exa}

\section{Codes over Eisenstein integers}

We construct codes being able to correct few restricted errors~$\eps$
where $\eps^6 = 1$.  Throughout we let $\rho \defeq \frac 1 2 ( 1
+ \sqrt 3 i) \in \C$ be a $6$-th primitive root of unity, so that
$\rho^3 = -1$ and $\rho^2 = \rho - 1$.  The \emph{Eisenstein integers}
are given as \[ \Z[\rho] \,=\, \Z + \Z \rho \,\subseteq\, \C \,, \]
and they form a unique factorisation domain and a hexagonal lattice in
$\C = \R + i \R \cong \R^2$.

\begin{fact} Every prime $p \equiv 1 \pmod 6$ is of the form
  $p = a^2 + a b + b^2$ with positive integers $a, b$, unique up to
  order. \end{fact}

On the Eisenstein integers we define the weight
\[ W(\alpha) \,\defeq\, \min \{ |x| + |y| + |z| \mid \alpha
  = x + \rho y + \rho^2 z \} \,. \]
This is the length of a shortest path from~$0$ to~$\alpha$ using steps
$\pm 1$, $\pm \rho$, $\pm \rho^2$.  We let $\pi \defeq a + \rho b \in
\Z[\rho]$.  The ideal $(\pi)$ can be seen as a sublattice with basis
$\pi, \rho \pi$ in the lattice $\Z[\rho]$ spanned by $1, \rho$, which
is of index $a^2 + a b + b^2 = p$, since
\[ \svec{ \pi \\ \rho \pi } \,=\, \begin{bmatrix} \,a & b\, \\
   -b\! & \!a \!+\! b \end{bmatrix} \svec{ 1 \\ \rho } . \]
We thus have $\F_p \cong \Z[\rho] / (\pi)$ and define the quotient
weight \[ w([\alpha]) \,\defeq\, \min \{ W(\alpha') \mid \alpha'
  \in [\alpha] \} \,, \] which in turn gives a “hexagonal” distance.

\begin{exa} For $p = 19$ we have $\pi = 3 + 2 \rho$ and the field
  $\F_p = \{ -9, \dots, 9 \}$ can be illustrated in the hexagonal
  distance as follows.
  \begin{center}\begin{tikzpicture}[scale=0.6]
      \node at (0, 0) {$0$};
      \node at (1, 0) {$1$}; \node at (-1, 0) {$-1$};
      \node at (2, 0) {$2$}; \node at (-2, 0) {$-2$};
      \node at (-1.5, 0.87) {$6$}; \node at (1.5, -0.87) {$-6$};
      \node at (-0.5, 0.87) {$7$}; \node at (0.5, -0.87) {$-7$};
      \node at (0.5, 0.87) {$8$}; \node at (-0.5, -0.87) {$-8$};
      \node at (1.5, 0.87) {$9$}; \node at (-1.5, -0.87) {$-9$};
      \node at (-1, 1.73) {$-5$}; \node at (1, -1.73) {$5$};
      \node at (0, 1.73) {$-4$}; \node at (0, -1.73) {$4$};
      \node at (1, 1.73) {$-3$}; \node at (-1, -1.73) {$3$};
    \end{tikzpicture}\end{center} \end{exa}

Since $W(\pi) = W(\rho \pi) = W(\rho^2 \pi) = a + b$ we easily see the
following.

\begin{fact} There holds $\min \{ W(\alpha) \mid 0 \ne \alpha \in
  (\pi) \} = a + b$. \end{fact}

Consider the “unit hexagon” $\Delta \defeq \{ \alpha \in \C \mid
W(\alpha) \le \frac 1 2 \}$, which is the convex hull of the points
$\frac 1 2 \{ \pm 1, \pm \rho, \pm \rho^2 \}$.  Then the fundamental
region $R \defeq \{ \alpha \in \Z[\rho] \mid W(\alpha) = w([\alpha]) \}$
is also given by $R = \Delta \pi \cap \Z[\rho]$.  The set~$R$ forms a
set of unique representatives for $\F_p \cong \Z[\rho] / (\pi)$.

\subsection{Code construction}

Choose a set $A \subseteq \Z$ with $A = -A$ and $A + \rho A \subseteq R$,
which specifies maps
\begin{gather*}
 \phi \colon A + \rho A \to A \,, \quad x + \rho y \mapsto x \\
 \psi \colon A + \rho A \to A \,, \quad x + \rho y \mapsto y
\end{gather*}
so that $\alpha = \phi(\alpha) + \rho \psi(\alpha)$ for all $\alpha \in
A + \rho A$.

Let $n \defeq |A|^2$ and $J = \{ \alpha_1, \dots, \alpha_n \} \defeq
A + \rho A$.  Define the linear code~$\mcC$ over~$\F_p$ by the
parity-check matrix \[ H \,\defeq\,
  \begin{bmatrix} 1 & \dots & 1 \\
    \phi(\alpha_1) & \dots & \phi(\alpha_n) \\
    \psi(\alpha_1) & \dots & \psi(\alpha_n) \\
    \alpha_1^2 & \dots & \alpha_n^2 \end{bmatrix} . \]
Assume that $a + b \ge 6$.  The same reasoning as in the proof of
Theorem~\ref{thm:first} establishes the following.

\begin{thm} The above code~$\mcC$ has minimum hexagonal distance
  $\ge 6$. \end{thm}

\begin{exa} In the case $p = 19$ we can use $A \defeq \{ -1, 0, 1 \}$.
  Then $n = 9$, $J = \{ -9, -8, -7, -1, 0, 1, 7, 8, 9 \}$ and the
  parity-check matrix is \[
    \begin{bmatrix} \,1 & 1 & 1 & 1 & 1 & 1 & 1 & 1 & 1\, \\
      -1\! & 0 & 1 & \!-1\! & 0 & 1 & \!-1\! & 0 & 1\, \\
      -1\! & \!-1\! & \!-1\! & 0 & 0 & 0 & 1 & 1 & 1\, \\
      \,5 & 7 & \!-8\! & 1 & 0 & 1 & \!-8\! & 7 & 5\, \end{bmatrix} . \]
  As $a + b = 5$ the argument shows that the minimum distance is
  $\ge 5$, so that two hexagonal errors can be corrected. \end{exa}

\subsection{Alternative construction}

Over the Eisenstein integers we present a different construction of
two-error correcting codes.  It does not rely on decomposing elements
into parts but rather on some kind of generalized Newton identity for
two errors.

\begin{lem} Let~$F$ be a field of characteristic $\ne 2, 7$ and let
  $\alpha, \beta \in F$.  With $S_1 \defeq \alpha + \beta$,
  $S_4 \defeq \alpha^4 + \beta^4$, $S_7 \defeq \alpha^7 + \beta^7$,
  if $S_1 \ne 0$, $S_1^4 + S_4 = 2 (\alpha^2 + \alpha \beta +
  \beta^2)^2 \ne 0$ there holds \[ \frac {2 (S_1^7 - S_7)}
    {7 S_1 (S_1^4 + S_4)} \,=\, \alpha \beta \,. \] \end{lem}

Considering $S_1, S_4, S_7$ as syndromes we are therefore able to
recover the error locator polynomial $(1 - \alpha z) (1 - \beta z) =
1 - (\alpha + \beta) z + \alpha \beta z^2$.

\begin{proof} We compute $ S_1^7 - S_7 = 7 \alpha \beta (\alpha^5 +
  3 \alpha^4 \beta + 5 \alpha^3 \beta^2 + 5 \alpha^2 \beta^3 +
  3 \alpha \beta^4 + \beta^5)$ and $S_1 (S_1^4 + S_4) =  2 (\alpha^5 +
  3 \alpha^4 \beta + 5 \alpha^3 \beta^2 + 5 \alpha^2 \beta^3 +
  3 \alpha \beta^4 + \beta^5$). \end{proof}

When decoding restricted errors, often not all the syndromes
$S_1, S_2, S_3, \dots$ are available but only a subset, for example
$S_1, S_3, S_5, \dots$ in the case of Berlekamp's negacyclic codes.
While the syndromes $S_1, S_3$ also recover the errors as
\[ \frac {S_1^3 - S_3} {3 S_1} \,=\, \alpha \beta \,, \]
we remark that $S_1, S_5, S_9$ seem not sufficient; in fact Huber's
algorithm~\cite{huber} uses $S_1, S_5, S_9, S_{13}$.

In the following we denote by $\zeta \defeq \rho^2$ the third
primitive root of unity.  Let $J = \{ \alpha_1, \dots, \alpha_n \}$ be
a set of representatives modulo multiplication by~$\zeta$ or~$\zeta^2$.
Each class in~$\F_p^*$ has three elements, hence including zero there
are \[ n \defeq \tfrac 1 3 (p \!-\! 1) + 1 = \tfrac 1 3 (p \!+\! 2) \]
classes.  If $\gamma \in \F_p^*$ is a primitive element, so that $\zeta
= \gamma^{(p - 1) / 3}$, we can take $J \defeq \{ 0, 1, \gamma, \dots,
\gamma^{(p - 4) / 3} \}$.  Consider the linear code~$\mcC$ over~$\F_p$
with parity-check matrix
\[ H \,\defeq\, \begin{bmatrix} \,1 & \dots & 1\, \\
    \,\alpha_1 & \dots & \alpha_n\, \\
    \,\alpha_1^4 & \dots & \alpha_n^4\, \\
    \,\alpha_1^7 & \dots & \alpha_n^7\, \end{bmatrix} . \]

We continue to assume that $a + b \ge 6$.

\begin{thm} The above code~$\mcC$ has minimum hexagonal distance
  $\ge 6$. \end{thm}

\begin{proof} As before, let $c \in \mcC$ be a codeword of weight
  $< 6$, and denote the rows of the parity-check matrix by
  $h^1, \dots, h^4$.  Since $a + b \ge 6$ we infer from
  $h^1 c^{\top} = 0$ that the error values sum up to~$0$.  This allows
  us to decompose the codeword as $c = c_+ - c_-$, where $c_+, c_-$
  have disjoint support with errors $1, \zeta, \zeta^2$ only.
  Here~$c_+$ accounts for the errors $+1$, $+\zeta$, $+\zeta^2$
  and~$c_-$ for the errors $-1$, $-\zeta$, $-\zeta^2$, thus covering
  all errors in $\pm 1, \pm \rho, \pm \rho^2$ since $-1 = \rho^3$.
  The codeword~$c$ is then of even weight~$\le 4$ and we have the
  following.

  The vector~$c_+$ has at most two errors, say at positions
  $\alpha, \beta$, with values $\zeta^e, \zeta^f$ for some
  $e, f \in \{ 0, 1, 2 \}$.  Then the vector~$c_-$ has these values at
  some positions $\alpha', \beta'$.  To simplify notation, let us
  assume a typical case of a $1$-error and a $\zeta$-error.  From the
  parity-check equations we obtain
  \begin{align*}
    \alpha + \zeta \beta \,&=\, \alpha' + \zeta \beta' \\
    \alpha^4 + (\zeta \beta)^4 \,&=\, \alpha'^4 + (\zeta \beta')^4 \\
    \alpha^7 + (\zeta \beta)^7 \,&=\, \alpha'^7 + (\zeta \beta')^7
  \end{align*}
  noting that $\zeta = \zeta^4 = \zeta^7$.  Now we deduce
  $\alpha (\zeta \beta) = \alpha' (\zeta \beta')$ and thus recover the
  set $\{ \alpha, \zeta \beta \} = \{ \alpha', \zeta \beta' \}$ by
  using the lemma.  (Indeed, $S_1 = 0$ does not hold by our choice
  of~$J$, and if $S_1^4 + S_4 = 0$ we obtain $\alpha (\zeta \beta)$
  directly.)  Finally, since the positions~$\alpha_i$ are chosen to be
  unique modulo multiplication by~$\zeta$ or~$\zeta^2$, we arrive at a
  contradiction unless $c = 0$. \end{proof}

\begin{exa} For $p = 37$ we let $n = 13$ and choose the
  parity-check matrix \[ \begin{bmatrix}
      \,1 &  1 &  1 &  1 &  1 &  1 &  1 &  1 &  1 &  1 &  1 &  1 &  1\, \\
      \,0 &  1 &  2 &  4 &  8 & \!16\! & \!-5\! & \!\!-10\!\! & \!17\! &
        \!-3\! & \!-6\! & \!\!-12\!\! & \!13 \\
      \,0 &  1 & \!16\! & \!-3\! & \!\!-11\!\! &  9 & \!-4\! & \!10\! &
        \!12\! & 7 &  1 & \!16\! & \!-3 \\
      \,0 &  1 & \!17\! & \!-7\! & \!-8\! & \!12\! & \!\!-18\!\! &
        \!\!-10\!\! & \!15\! & \!-4\! &  6 & \!-9\! & \!-5 \\
    \end{bmatrix} \]
  for a code of minimum distance $\ge 6$, thus correcting two
  hexagonal errors. \end{exa}

\begin{rem} We can also choose the $\alpha_1, \dots, \alpha_n$ from an
  extension field of~$\F_p$ to obtain two-error correcting codes of
  larger length.  For example, using the finite field~$\F_{p^2}$ we
  get such codes of length $n \defeq \frac 1 3 (p^2 + 2)$ with a
  parity-check matrix over~$\F_p$ having~$7$ rows. \end{rem}

\section*{Conclusion}

In this work we propose novel constructions of codes correcting few
restricted errors, where we focus on the case of $m = 4$ or $m = 6$
error values.  These codes have larger minimum restricted weight than
corresponding MDS codes would have for the Hamming weight.  Over the
Gaussian integers our codes tend to have better parameters when
compared to Huber's icyclic codes~\cite{huber}.

In principle, the first, more geometric code constructions may be
extended to correct four or more errors, if a larger prime~$p$ is
chosen and the set~$A$ is made smaller.  It is an interesting problem
to extend the results to other subgroups of restricted errors beyond
$m > 6$.  However, the cyclotomic rings $\Z[\zeta_m]$ for $\zeta_m
\defeq \exp ( \frac{2 \pi i} m ) \in \C$ do not embed into a
two-dimensional lattice in this case.

It would also be fascinating to extend the Newton type identities for
other sets of syndromes, either to correct more than two errors or to
handle larger subgroups~$E$.

Finally, we remark that codes of larger restricted distance with
possible applications in a code-based McEliece like cryptosystem may
be obtained by applying a product code construction~\cite{freuden1}.

\end{document}